\documentclass[letterpaper,10pt,conference]{IEEEtran}
\usepackage{setspace}
\usepackage{amsmath,amsthm, amssymb}
\usepackage{times}
\usepackage{cite}
\usepackage{graphicx}
\usepackage[mathscr]{euscript}
\usepackage[ruled,lined]{algorithm2e}
\usepackage{algorithmicx}
\usepackage{xcolor}
\usepackage[compatible]{algpseudocode} 
\usepackage{geometry}{}
\usepackage{subcaption}
\usepackage{url}
\DeclareMathOperator*{\minimize}{minimize}
\newtheorem{definition}{Definition}
\newtheorem{exmp}{Example}

\newtheorem{theorem}{Theorem}
\newtheorem{lemma}{Lemma}

\IEEEoverridecommandlockouts

\begin{document} 
\title{Optimal Storage Aware Caching Policies for Content-Centric Clouds}
\author{\IEEEauthorblockN{Samta Shukla and Alhussein A. Abouzeid}\\
\IEEEauthorblockA{Department of Electronic, Computers and Systems Engineering\\
Rensselaer Polytechnic Institute, USA\\
\{shukls, abouzeid\}@rpi.edu}}
\singlespacing
\setstretch{1.2}

\maketitle

\begin{abstract}
\par Caches in Content-Centric Networks (CCN) are increasingly adopting flash memory based storage. The current flash cache technology stores all files with the largest possible ``expiry date," i.e. the files are written in the memory so that they are retained for as long as possible. This, however, does not leverage the CCN data characteristics where content is typically short-lived and has a distinct popularity profile. Writing files in a cache using the longest retention time damages the memory device thus reducing its lifetime. However, writing using a small retention time can increase the content retrieval delay, since, at the time a file is requested, the file may already have been expired from the memory. This motivates us to consider a joint optimization wherein we obtain optimal policies for jointly minimizing the content retrieval delay (which is a network-centric objective) and the flash damage (which is a device-centric objective). Caching decisions now not only involve \emph{what} to cache but also for \emph{how long} to cache each file. We design provably optimal policies and numerically compare them against prior policies.
\end{abstract}


\section{Introduction}
\label{sec:intro}
%\input{introduction.tex}
%\subsection{Caching in Data Centers of Content-Centric-Network}

\par Flash memory based cache is a principal component of the emerging Content-Centric or Information-Centric Networks (CCN/ ICN) with ubiquitous caching, mobile/cloud computing and device-to-device networking \cite{Guo:2009,Ko:2012,Armbrust:2010,6232368}. One of the main obstacles in flash memory adoption is its high rate of wear out (referred as flash damage) which is directly proportional to the programmed data retention time \cite{6232368,reconfs, wear, lifetime, microsoft_presentation}.

\par  The relationship between data retention and wear out can be briefly described as follows. A flash memory consists of flash cells. Data is stored in a flash memory by \emph{programming} (P) the threshold voltage of each memory cell into two or more non-overlapping voltage windows. A memory cell is \emph{erased} (E) of all the data before it is programmed; erasing data involves removing the charges in the floating gate and setting the threshold voltage to the lowest voltage window. The reliability of a flash (or flash lifetime) is specified in terms of the number of program/erase (P/E) cycles it can endure (e.g., $10^4$ to $10^5$ P/E cycles) \cite{desnoyers}. 
Depending on the underlying technology, all flash cells are programmed to retain data in cache for a specified duration (from 1 to 10 years), known as the \emph{retention time}.
The specified memory retention is achieved by programming data with a high threshold voltage. However, programming at high voltages causes a high wear out to the flash cell thus reducing the memory lifetime \cite{reconfs, wear, lifetime}.
 
\par The current practice in flash technology is not optimal. It writes all files with a fixed maximum/default threshold voltage to get the maximum possible data retention which in turn causes maximum cache damage at each write. Note that the high damage caused is permanent even if the file is evicted from the cache before its retention expires. Clearly, writing every content with maximum retention is wasteful since in CCN some content could be less popular. We aim to obtain optimal retention times by leveraging the content popularity profile which can be locally estimated from the user requests in CCN architectures \cite{siris2014efficient}.

%\subsubsection*{Contributions}
\par We take first steps in reformulating the traditional data caching problem by proposing a \emph{cross-layer optimization} that combines the \emph{cache-level} objective of minimizing the content-retrieval delay and the \emph{device-level} objective of minimizing the device damage\footnote{A preliminary version of this work appeared in \cite{Samta_WiOpt_2016}. This paper extends the work in \cite{Samta_WiOpt_2016} with: (1) A complete description of the online policy in Section \ref{sec:complete_policy}. (2) Section \ref{sec:policyperformance}, wherein we compare the online and offline policies by showing the delay-damage tradeoffs and the competitive ratios. (3) The full proofs for all theorems and lemmas in the Appendices.}. We note that the cache-level and device-level objectives are conflicting: a smaller delay is achieved by writing files with longer retention but that incurs a high damage; a smaller (or zero) damage is achieved by not writing files at all but that causes large delays. Despite the inherent trade-off, earlier work in these areas have progressed largely independently. For example, in the device literature, a recent line of work considers optimizing damage/retention times by dynamically trimming the retention duration of a content based on their refresh cycle durations \cite{liu2012optimizing, retentiontrimming}. Another closely related work considers trading retention time for system performance (such as memory speed and lifetime) \cite{microsoft_presentation}.
By contrast, the caching literature consists of innumerable attempts to construct policies with high cache hit probability for achieving lower network delays (see \cite{fofack2014models} and the citations within) while overlooking the device-related aspects.

The key challenge addressed in this paper is to find caching policies for a finite capacity cache, that, in addition to the functions provided by a traditional caching policy, determine optimal file retention times 
to incur minimum flash damage when subject to a constraint on acceptable network delay. A file written in the cache at time $t$ for a retention duration $D$ is no longer readable from the cache after time $t+D$, thus leading to a cache miss (unless the file is re-written between $t$ and $t+D$ to extend its original retention but that incurs additional damage). 

\par Our first contribution (Section \ref{sec:optoffpol}) is to solve the problem of \emph{offline caching}, i.e. caching when the content request string is given. We design an optimal offline policy, Damage-Aware REtention (DARE), that returns the optimal retention times for every file without exceeding the optimal cache misses given by Belady's Farthest in Future (FiF) algorithm\footnote{With traditional caching (caching without optimizing memory retentions) Belady's Farthest in Future (FiF) algorithm obtains the optimal cache misses. As FiF is damage/retention agnostic, we account for the flash damage in FiF by assigning the retention time of a file as the duration for which it stays in the cache before it is evicted due to a cache miss.}. 

We prove analytically and show by simulations that our policy, DARE, by taking retentions into account, achieves a significantly lower cache damage than FiF without increasing the optimal delay (or cache misses). 

\par Our second contribution is to solve the \emph{online caching} problem, i.e. caching when the request string is not given ahead of time. Our optimal online policy, DARE-$\Delta$, approaches the online caching problem in two stages. It first assumes a large cache (a cache with no capacity constraint) and obtains the optimal file retentions by solving an optimization problem (Section \ref{sec:optonpol}). A large cache assumption implies that there are no evictions and the cache misses are only because of the files expiring. The policy then extends the results from a large cache case to a cache of finite capacity in Section \ref{sec:fincache}. In this case a cache miss can result in a file eviction if the cache is full (Section \ref{sec:fincache}). Subsequently, DARE-$\Delta$ exploits the optimal retentions obtained for large caches and models the problem of which file to evict at every cache miss as a  Markov Decision Process (MDP). In contrast with the usual MDP-based approaches which suffer from the curse of dimensionality, we show that our MDP can be characterized to give a very simple, easy to implement rule for evicting files.
Our simulations (Sections \ref {sec:optonpol}, \ref{sec:fincache}, \ref{sec:policyperformance}) reinforce the theoretical findings for a range of parameters, caching policies and damage functions for the online case. We note that our work is a significant generalization of \cite{omri} where authors found an eviction sequence using MDP but do not consider flash damage constraints in their formulation.

\section{Preliminaries}
\label{sec:prelim}
In this section, we explain the model assumptions that are common to all the analytical results in the paper.

We also discuss our work in the light of closely related literature.
\subsection{Model Assumptions}
\subsubsection{Cache-level assumptions}
Our model for online caching is based on the following model assumptions. The file arrivals conform to the Independent Reference Model (IRM)\footnote{Although IRM does not take temporal locality into account, it is a widely accepted, standard traffic model in caching literature \cite{omri,martina}.} \cite{omri,martina}, where each file is requested with a static probability independent of other requests. We describe our traffic model in more detail in Section \ref{subsec: on_model}. For tractability, we obtain results for Poisson file arrivals modulated with a suitable popularity distribution -- such as, ZipF popularity law \cite{omri,martina} -- and exponentially distributed retention times in our analysis in Sections \ref{sec:optonpol} and \ref{sec:fincache}\footnote{Our Markovian formulations in Section \ref{sec:fincache} require memoryless arrivals and retention times.}. For ease of exposition, we assume that files are fetched from the server (upon a cache miss) by incurring deterministic delays. This implies that the delay minimizing objective translates to minimizing the number of cache misses. Thus we will use minimizing delay and minimizing cache misses interchangeably in the rest of the paper.

Let $M$ denote the set of all files where each file $m\in M$ is of unit size\footnote{Our results can be generalized to account for file sizes, we adopt unit file sizes for ease of exposition.}. Files are requested at a cache with finite capacity of size $B$ files. A requested file that is not in the cache results in a cache miss. Upon a cache miss, the requested file is fetched by incurring a delay cost (see Section \ref{sec:cost_jhamela}) and is subsequently written in the cache by incurring a retention cost (see Section \ref{sec:cost_jhamela}). Files are served instantaneously in the case of a cache hit. 

\subsubsection{Device-level assumptions}
\par The process of writing files in the flash cache is explained as follows. A memory is divided into various sectors from which a sector is chosen uniformly at random. It is a reasonable assumption since the disk controller in a flash exercises ``wear leveling" by spreading writes evenly across the flash chip for causing less damage to the flash lifetime \cite{Koltsidas:2011}. We neglect the damage caused due to subsequent reads of an already written file and only consider the damage due to writing a file since reading the disk does not require writing or erasing \cite{Koltsidas:2011}. Subsequently, we model the P/E cycle counts and erasure costs (associated with programming and erasing a file) in a flash memory with the help of a damage function which takes retention times as arguments (see Section \ref{sec: sys_model_off}). This is justified because the P/E cycle duration is closely related to the retention time. A higher retention is obtained by programming (P) the flash with a very high positive voltage thus requiring a very high negative voltage to erase (E) the data. Finally, while there are only empirical relationships known about flash damage as a function of the depleting cell life \cite{lifetime}, we propose and analyze a general mathematical model that captures a wider range of dependence between flash damage and depleting cell life due to file retention (see Section \ref{sec: sys_model_off}).

\subsection{Cost of fetching and writing a file}
\label{sec:cost_jhamela}
\par The total cost of fetching and writing file $m$ upon a cache miss consists of a delay cost $\delta(m)$ and retention (writing) cost $f(m)$. The total cost per miss of file $m$ is denoted by $c(m)=\delta(m)+f(m)$.

%\begin{itemize}
\textit{Delay cost $\delta(\cdot)$}:
For every cache miss, fetching the requested file from the server results in a determinimistic delay cost which can be thought of as the transmission delay to obtain the file from the server based on the time of the day, current server workload, or available channel bandwidth, etc.

\textit{Retention cost $f(\cdot)$}: Retention cost is incurred due to flash memory damage. While there are only empirical relationships known about flash memory damage as a function of memory retention times, we outline two desirable properties for constructing a suitable damage function: (1) Memory damage, although a function of several factors, is known to increase with retention time; this is because writing a file at a higher threshold voltage helps in a longer file retention thereby incurring a higher damage \cite{reconfs, wear, lifetime}; (2) Damage function, $f(\cdot)$, is a complicated, non-linear function with $f(0)=0$. Based on these properties, we choose a \emph{convex increasing polynomial} as a damage function satisfying both (1)-(2). 

\par The total cost descriptions for offline and online policies are discussed in Sections \ref{sec: sys_model_off} and \ref{subsec: on_model}, respectively.

\subsection{DARE caches vs. TTL caches}
\par Having a file written for a duration equal to its retention time, as in DARE caches, may appear similar to the TTL caches considered in \cite{fofack2014models, hari, berger, towsley}, where files stay in cache for their TTL (time-to-live) duration. 
 
However, our work, even at the conceptual level, is different from TTL caches\footnote{Coincidentally, the hit and miss probabilities obtained for DARE with the large cache assumption are the same as the hit and miss probabilities of a TTL cache under the RND caching policy (see Section \ref{subsec:caching_policy}).}: (1) DARE considers both finite and infinite capacity caches whereas TTL considered infinite capacity caches only. Analyzing a finite capacity cache is paticularly applicable for CCN routers which are known to have small caches \cite{fofack2014models}. (2) The goal of DARE caching is to minimize flash damage with acceptable delay guarantees. DARE takes retention time distributions as input and outputs the optimal retention values satisfying the goal. In contrast, TTL caching is a modeling technique devised to simplify the analysis of traditional caching policies. They take a damage oblivious existing policy as input to obtain (an asymptotic approximation of) the corresponding TTL distribution as an output (see \cite{fofack2014models} for a detailed analysis of TTL caches).
\subsection{Summary of prior caching policies}
\label{subsec:caching_policy}
We compare our optimal policies against the performance of the following well-known policies (e.g. see \cite{fofack2014models}). In these policies, a requested file not already in the cache is inserted. The policies differ in their eviction policies when a cache is full. In Least Recently Used (LRU) policy,  the least recently used file is evicted. In First In First Out (FIFO) policy, the file which was written first is evicted. In  RaNDom (RND) policy a file is evicted from the cache uniformly at random. Farthest in Future (FiF) policy, also called Belady's Algorithm, evicts the file whose next request is the farthest in time. FiF minimizes the number of cache misses \cite{belady} but assumes knowledge of the full time sequence of requests. LRU is widely used since 
it performs well even for arbitrary request strings. RND and FIFO are very simple to implement in hardware and are seen as a viable alternative of LRU in CCN high-speed routers \cite{martina}.
% %
%%

\section{Flash-aware Optimal Offline Caching}
\label{sec:optoffpol}
In this section, we consider the case of offline caching. In this case, the file request string is given ahead if time as a sequence of positive integer-valued indices chosen from a set of $M$ files.  
Recall the FiF algorithm by Belady \cite{belady} which is known to minimize the number of request misses for a cache. Our contribution is in showing that FiF is not optimal with respect to damage. Further, we advance the state-of-the-art by constructing the $\text{DARE}$ caching policy which minimizes flash damage by taking no more delay (cache-misses) than Belady's FiF (i.e. the known optimal delay benchmark).

\subsection{System model}
\label{sec: sys_model_off}
In this section, we assume that time of horizon length $T$ is slotted in equal length intervals, and files are requested at the beginning of each slot. A requested file that is not in the cache results in a cache miss. Upon a cache miss, the requested file is fetched and subsequently written in the cache for \emph{at least one slot}. Writing the requested file on \emph{every} miss is called \emph{cache miss allocation} \cite{Pritchett} in device literature. We lift this assumption in Section \ref{sec:policyperformance} where the policy is allowed to skip writing the requested file.

\par The total cost of fetching and writing file $m$ upon a cache miss consists of a delay cost $\delta(m)\in \mathbb{Z}^+$ and retention (writing) cost $f(m)\in \mathbb{Z}^+$ as defined in Section \ref{sec:cost_jhamela}. We assume that the delay cost $\delta(m)=1$ unit for all $m\in M$. With this assumption minimizing delay corresponds to minimizing the number of cache misses. 
%We adopt this is for the ease of illustration and to lend insights to the case of online caching. 
Let the one-shot retention cost caused due to writing file $m\in M$ for a retention time $R\geq 1$ slots, $R\in \mathbb{Z}^+$ be an increasing, convex function given by $f(R)\in \mathbb{Z}^+$. Thus, the total cost is given by the sum of one-shot delay and retention costs for every slot in the horizon corresponding to a cache miss, i.e. $\text{offline cost}=\sum_{t=1}^{T}{1_{(m,t)}(1+f(R_m))}$, where $1_{(m,t)}=1$ if there was a cache miss on file $m$ at time $t$ and 0 otherwise. 

\par Let $F, E$ denote the optimal number of cache misses, the corresponding eviction sequence according to FiF policy. Our goal is to find a policy that determines the \emph{optimal retention times} for each file write without exceeding $F$.

\subsection{The optimal offline policy, DARE}
%\par We find the optimal policy, DARE, in Algorithm \ref{algo:off}. 
\par DARE aims to reduce the cache retention times without changing the cache miss sequence from FiF. It considers every eviction in the optimal eviction sequence given by FiF policy and works backward to find the optimal retention for \emph{each} file write. When a file $l$ is evicted in FiF at time $t$, DARE finds two different time indices by traversing back from $t$. First, it finds the \emph{latest} (time) slot when $l$ was written in the cache before getting evicted at $t$; we call it time $k$. Second, it searches for the time when $l$ was last requested before eviction at $t$, we call it time $j$. Our policy stores file $l$ in the cache at time $k$ for $j-k+1$ slots. Also, the files which are present in the cache (i.e. not evicted) till the last eviction are taken care of similarly. Thus, for each evicted file, DARE saves on the number of slots by storing a file for a retention time equal to the difference between the time when it was last requested from the time when it was written latest. Example~\ref{exmp:example} illustrates the algorithm.

\begin{exmp}
\label{exmp:example}
\normalfont
Consider a cache of size $B=3$ containing files $\{a,b,c\}$ at time $t=0$ with the request string in Table~\ref{tab:table}. 
%For simplicity, further assume that $\alpha=1$, $\beta=0$ in (\ref{eqn:cost}). 
\begin{table}[htbp]
\centering
\caption{Sequence of evictions and cache evolution with each request under DARE.}
\begin{tabular}{|l|l|l|l|} \hline
\textbf{Slot}&\textbf{Request}& \textbf{File evicted} & \textbf{Files in cache}\\
\hline
 1 & a & - & \{a,b,c\}\\
 \hline
% 2 & b & - & -do-\\
% \hline
 2 & e & b & \{a,c,e\}\\
 \hline
% 4 & a &- & -do-\\
% \hline
 3 & c & - & -do-\\
 \hline
 4 & a & - &-do-\\
 \hline
 5 & d & c & \{a,d,e\}\\
 \hline
 6 & a & - & -do-\\
 \hline
 7 & b & d & \{a,b,e\} \\
 \hline
 8 & e & - & -do-\\
 \hline
 9 & a & - & -do-\\
 \hline
\end{tabular}
\label{tab:table}
\end{table}
\par For each eviction, DARE calculates the retention time backwards. Consider slot 5 when a request for file $d$ results in a miss, and file $c$ is evicted as per the solution of FiF. We find the last time when $c$ was requested, i.e. $j-1=3$, i.e. $j=4$.  Note that, file $c$ was in cache starting from time $t=0$. Hence, file $c$ will be written for time $j-k=4-0=4$ slots. Similarly, it is easy to see that the output from DARE is to write both files $a$ and $e$ for 9 slots (since, files $a$ and $e$ are never evicted) and files $b$, $d$ for 1 slot each.
\end{exmp}

\subsection{DARE is optimal}
\par We observe that DARE incurs optimal number of cache misses (by definition). Thus, for optimality we only need to prove the non-existence of a policy which incurs less cost than DARE in choosing retention times for files without exceeding the optimal number of cache misses. 
%We report the optimality proof in the Appendix.

\begin{theorem}
\label{thm:opt_off}
DARE is optimal with respect to retention cost over all possible optimal eviction sequences that minimize the number of cache misses.
\end{theorem}
\begin{proof}
See Appendix B.
\end{proof}
\subsection{Numerical study: Cache miss versus damage}
The current practice in flash memory technology is to write all files with a very high retention (typically 1-10 years), however, for making a fair comparison among policies we assume that the policies LRU, FIFO, RND and FiF write a file exactly for the time till it is not evicted. Subsequently, we write a file for the calculated optimal retention duration for DARE. Our goal is to demonstrate the optimality of DARE against other caching policies.

%\subsubsection{Varying cache sizes}
\par We consider a horizon of length $T=10000$ slots where in each slot file $m$ is requested as per the IRM with probability $\frac{1/m^\alpha}{\sum_{j\in M}{{1/j^\alpha}}}, m\in M$, where $\alpha$ is the ZipF popularity coefficient. Usually, for web caches and data servers, the ZipF coefficient is found to vary from 0.65 (least skewed) to 1 (most skewed) \cite{fofack2014models}. Hence, we consider two extremes and set $\alpha=\{0.65, 0.95\}$. Files are requested from a catalogue containing $M=1000$ files and the cache size varies from 50 to 600 files. The damage function for writing files is assumed to be quadratic in retention time. We compute the aggregate damage and cache misses for $T$ slots by evaluating: $\text{damage}=\sum_{t=1}^{T}{1_{(m,t)}R_m^2}$, and $\text{cache miss fraction} = 1/T\sum_{t=1}^{T}{1_{(m,t)}}$, where $1_{(m,t)}=1$ when there is a cache miss for file $m$ at time $t$ and 0 otherwise. We plot the results in Figures \ref{fig:varcache1}, \ref{fig:varcache2} for $\alpha=0.65$ and $0.95$, repectively.

\par Recall that the cache misses (fraction) for both FiF and DARE are the same (by definition). Thus, it suffices to represent the cache miss variation by plotting a single curve, which is shown by the dotted curve in Figure \ref{fig:subvarcache2}. We observe that as the cache size increases, the fraction of cache misses decreases, as expected, and soon converges to a specific value in steady state. The higher the ZipF-$\alpha$, the sooner this fraction converges. We also note that a higher $\alpha$ results in a lower value of cache miss (fraction) in steady state. This can be briefly explained as follows. When $\alpha$ increases, the skewness in the file request arrivals increases, i.e. with $\alpha=0.95$ the popular files are more popular and the unpopular files are less popular, compared to $\alpha=0.65$. Thus, a highly skewed traffic, by sending fewer requests for unpopular files, begets a lower cache miss count. 
\par The solid lines in Figure~\ref{fig:subvarcache2} show that as the cache size increases, the damage values from both FiF and DARE increase, and gradually both of them converge to a specific value. This implies that the damage savings obtained, calculated as $\frac{\text{damage from FiF}}{\text{damage from DARE}}$ approaches one with the increase in cache size. We observe that for smaller caches, a damage savings of upto 2-3 folds can be achieved. We also compare DARE against LRU, FIFO and RND; simulating these policies result in significantly worse damage to the extent that it can not be shown on the figures with the same scale. Similar trends for the ZipF variation follow for the damage curve as observed for the cache miss (fractions) curve.

\begin{figure}[t]
\centering
\begin{subfigure}{.25\textwidth}
  \centering
  \includegraphics[width=.87\linewidth]{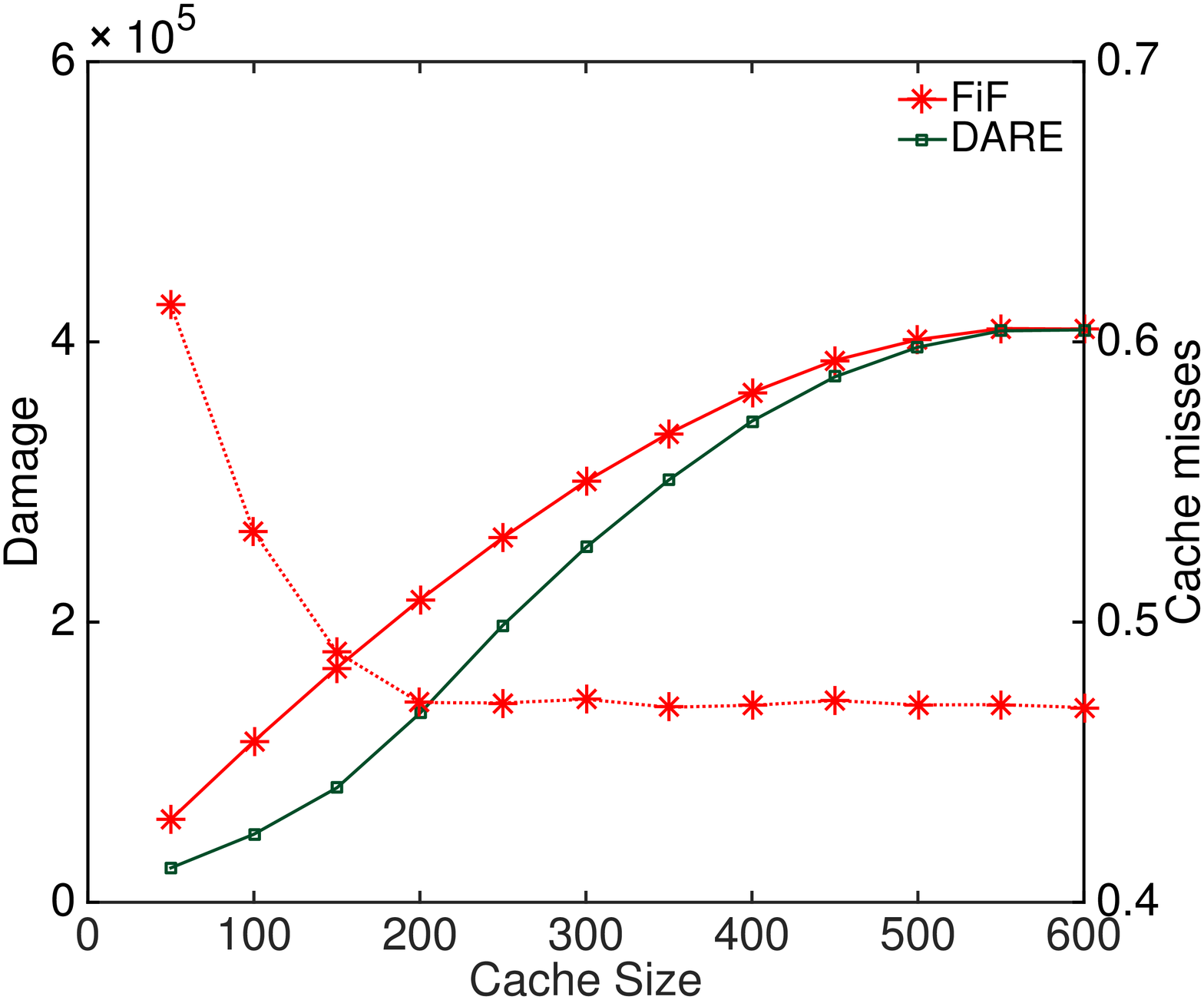}
  \caption{ZipF-0.65 popularity}
  \label{fig:varcache1}
\end{subfigure}%
~
\begin{subfigure}{.25\textwidth}
  \centering
  \includegraphics[width=.87\linewidth]{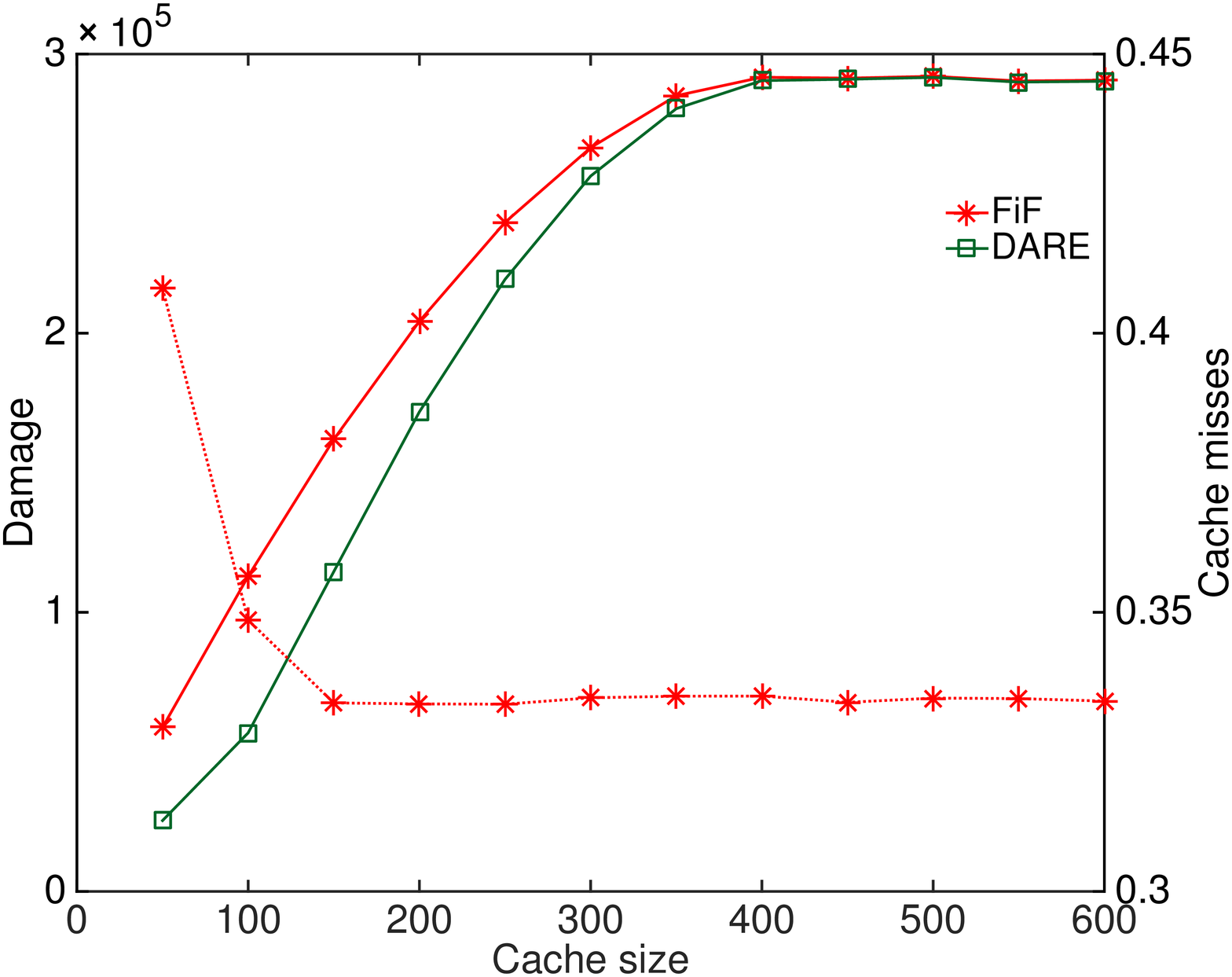}
  \caption{ZipF-0.95 popularity}
  \label{fig:varcache2}
\end{subfigure}
\caption{Cache damage vs. cache miss under IRM for B varying from 50 to 600 files with $|M|=1000$ files.}
\label{fig:subvarcache2}
\end{figure}

\par In this section, we showed that the well-known delay optimal caching policy (FiF) is not damage optimal. Further, we devised a caching policy that achieves optimal damage without exceeding the optimal number of cache misses given by the FiF policy. 
The case of offline caching lends insights to motivate the online caching problem where the arrival requests are not know apriori.
% % 

%%
\section{Flash-aware Optimal Online Caching for Large Caches}
\label{sec:optonpol}
\par We now consider the case of online caching where the files are requested according to a distribution, however, the exact request string is not known to the policy apriori. We first state the system model for the online caching which applies to Sections \ref{sec:optonpol} and \ref{sec:fincache}. Our goal is to design a policy that jointly finds the optimal retention times for all files and the optimal eviction sequence in the event of a cache miss. We achieve this goal by designing a policy DARE-$\Delta$ which optimizes in two steps. First, in this section (Section \ref{sec:optonpol}), it approximates the problem by considering a \emph{large cache} (a cache with no capacity constraint and hence no evictions) and finds the optimal retention times. Subsequently, in Section \ref{sec:fincache}, it obtains the optimal file eviction sequence given the optimal retention durations. Note that the problem of jointly optimizing over all possible retention times and eviction decisions remains an open problem.
\par In this section, we formulate an optimization problem called DARE-$\Delta$ Retention Formulation (see Section \ref{damage_formulation}) to minimize cache damage subject to a constraint on the network delay to find optimal retention times. Our formulation provides an approximate solution due to the large cache assumption, however, our numerical studies in Section \ref{sec:policyperformance} show that the objective function quickly converges to steady state with increasing cache size. Having a large cache implies that there are no evictions and there is a cache miss on the requested file only if it has expired from the cache; this assumption\footnote{A large cache assumption was previously considered in \cite{fofack2014models, hari} in the context of TTL-caches.} is known to decouple files thus facilitating a tractable mathematical analysis \cite{fofack2014models, hari}. Finally, we conclude this section by illustrating damage-delay trade-offs for different damage functions.

\subsection{System model}
\label{subsec: on_model}
\subsubsection{Traffic model}
\par The file request string is assumed i.i.d. File requests arrive according to the Independent Reference Model (IRM) \cite{omri}, \cite{martina} which assumes the following. (1) All requests are for a fixed collection of $M$ files. (2) The probability of requesting file $m$ is $p_m$ which is static and independent of past or future requests. 

\par We assume that the interarrival times of file $m\in M$, $X(m)$, is exponentially distributed with rate parameter $\lambda_m$, and the arrival process across files conform to an independent and homogeneous Poisson process. Under IRM, the probability of requesting file $m$ with interarrival times $X(m)$ and modulated with ZipF-$\alpha$ popularity law is given by $p_m=\frac{\lambda_m}{\sum_{j=1}^{M}{\lambda_j}}$ where $\lambda_m=1/m^\alpha, \; \forall m\in M$.

\subsubsection{Cost of fetching and writing a file}
\par The total cost of fetching and writing file $m$ upon a cache miss consists of a delay cost $\delta(m)\in \mathbb{R}^+$ and a retention (writing) cost $f(m) \in \mathbb{R}^+$ as defined in Section \ref{sec:cost_jhamela}. The retention time for file $m$ is assumed to be distributed as an exponential random variable $\mathscr{R}(m)$ with parameter $\mu_m$, $m\in M$ to (1) keep the problem tractable and (2) to capture the property that writing a file in memory with a retention $R$ leaves a non-zero probabilty of finding it in cache after time $R$.
\par In the event of a cache miss, a one-shot retention cost is incurred (see Definition \ref{def:damage}). The cumulative \emph{retention cost} is defined as the sum of all one-shot retention costs.

\begin{definition}[One-shot retention cost]
The one-shot retention cost is the damage caused to the cache due to writing a file for a retention time $Z\in \mathbb{R}^+$, given by $f(Z) \in \mathbb{R}^+$ where $f(\cdot)$ is a convex increasing polynomial of degree $n$ given by $f(Z)=a_nZ^n+a_{n-1}Z^{n-1}+\dots+a_1Z+a_0$ with coefficients $a_i\geq 0$, for all $i\geq 1$ and $a_0=0$.
\label{def:damage}
\end{definition}

\subsection{Problem formulation for finding optimal retention times}
\label{damage_formulation}
\begin{definition}[Optimal online policy]
\label{def:optonpol}
A policy is online optimal if it finds the values of the retention parameters for each file (i.e. $\{\mu_m\}$) that minimizes the expected cache damage due to successive file writes under the constraint that the expected delay does not exceed $\Delta>0$.
\end{definition}

To find the optimal online policy (see Definition \ref{def:optonpol}), we first obtain an expression for the miss probability with a single file in the library ($|M|=1$) and consider the set of all requests to a cache in steady state. Let $\{\mathscr{R}_n\}$ denote\footnote{We denote the discrete retention time in the offline caching section as $R$ and the continuous retention for the context of online caching as $\mathscr{R}$.} the i.i.d. exponential retention time sequences corresponding to arrivals $n=1,2,\dots$ for the single file. Let $I_n$ be the indicator variable defined as follows:

$I_n = \begin{cases} 1 &\mbox{if } \text{ $n^{th}$ file request results in a cache miss } \\ 
0 & \text{ otherwise }\\
 \end{cases}$
 
\par Let $X_n{(m)}$ denote the i.i.d exponential interarrival time between the $n^{th}$ and ${n+1}^{th}$ request of file $m$. Note that $I_n=1$ corresponds to the event $X_n>\mathscr{R}_n$. Thus, $\lim_{N\rightarrow \infty} \frac{1}{N}\sum_{n=1}^{N}I_n=\mathbb{P}(X_n>\mathscr{R}_n)=p_{\text{miss}}=\frac{\mu}{\lambda+\mu}$. Similarly, the probability of a cache hit is, $p_{\text{hit}}= 1-p_{\text{miss}}=\frac{\lambda}{\lambda+\mu}$. For the $n^{th}$ file request, we write the file with retention $\mathscr{R}_{n+1}$ if there is a miss (and we do not write otherwise). Thus, the \emph{expected damage}, $D$, can be expressed as:
\begin{align}
D&= \lim_{N\rightarrow \infty}\left[\mathbb{E}_\mathscr{R} \left[ \frac{1}{N}\sum_{n=1}^{N}I_n\times f(\mathscr{R}_{n+1})\right]\right] \nonumber\\
&= \lim_{N\rightarrow \infty} \frac{1}{N}\sum_{n=1}^{N}I_n\times \mathbb{E}_\mathscr{R}\left[{f(\mathscr{R}_{n+1})}\right],\\
&= p_{\text{miss}}\times {\mathbb{E}_\mathscr{R}\left[{f(\mathscr{R}_{n+1})}]\right]}
\label{damage}
\end{align}
%\begin{exmp}
%\normalfont
\noindent
which is true since $I_n$ is independent of the retention time $\mathscr{R}_{n+1}$\footnote{$I_n$ only depends on $X_n$ and $\mathscr{R}_n$ by definition.}. Similarly the expected delay constraint can be expressed as $p_{\text{miss}}\times \delta\leq \Delta $.
When $\mathscr{R}_{n+1} \sim exp(\mu)$ and $f(x)=x^2, x\geq 0$, then $\mathbb{E}_\mathscr{R}[\mathscr{R}_{n+1}]=2/\mu^2$, which is independent of $n$. Thus, for a single file, the goal is to minimize $p_{\text{miss}}\times \frac{2}{\mu^2}$ subject to the constraint $p_{\text{miss}}\times \delta \leq \Delta$.
\par We generalize the formulation obtained for a single file to the set of $|M|$ files. With IRM, the probability of requesting file $m$ is given by $p_m$, where $p_m=\lambda_m/\sum_{i\in M}{\lambda_i}, m\in M$. Also, the miss probability of file $m$ upon request is given by,
$p_{\text{miss}}(m)={\mathbb{P}(X(m)>\mathscr{R}(m))}= {\mu_m}/{(\mu_m+\lambda_m)}$, since the interarrival and retention times are exponentially distributed.
Thus the optimization problem becomes:
\begin{subequations}
\label{obj-dam}
\begin{align}
 &\minimize_{\mu_m\in \boldsymbol{\mu}} \sum_{m\in M}{{p_{\text{miss}}(m)}\times p_m}\times \mathbb{E}_\mathscr{R}\left[{f(\mathscr{R}(m)}\right] \\
 &\text{ subject to  }  \sum_{m\in M}{{{p_{\text{miss}}(m)}\times p_m}\times \delta(m)}\leq \Delta
 \end{align}
 \end{subequations}
% \begin{align}
% &\min \sum_{j=1}^{M}{\frac{\mu_j}{\mu_j+\lambda_j}\times \frac{\lambda_j}{\sum_{i=1}^{M}{\lambda_i}}}\times \frac{1}{\mu_j} \nonumber \\
% &\text{ subject to  }  \sum_{j=1}^{M}{\frac{\mu_j}{\mu_j+\lambda_j}\times \frac{\lambda_j}{\sum_{i=1}^{M}{\lambda_i}}}\leq \epsilon
% \end{align}

Define $q_m:={\lambda_m}/{(\mu_m+\lambda_m)}, m\in M$, and substitute the value of the polynomial damage function, $f(x)=a_nx^n+a_{n-1}x^{n-1}+\dots+a_1x$ (as defined in (\ref{def:damage})) in the objective of formulation (\ref{obj-dam}). The objective becomes:
\begin{small}
\begin{align}
\frac{1}{\sum_{m\in M}{\lambda_m}}\sum_{m\in M}{q_m \mu_m \mathbb{E}[a_n\mathscr{R}(m)^n+\dots+a_1\mathscr{R}(m)]}\label{eqn:5}.
\end{align}
\end{small}
Note that $\mathbb{E}[a_k\mathscr{R}(m)^k]=a_kk!/\mu_m^{k}$ for $\mathscr{R}\sim exp(\mu)$. Also,
\begin{small}
\begin{align}
\frac{1}{\mu_m^k}=\frac{1}{(\mu_m+\lambda_m-\lambda_m)^k} = \frac{1}{(\frac{\lambda_m}{q_m}-\lambda_m)^k}= \left(\frac{q_m/\lambda_m}{1-q_m}\right)^k. \label{eqn:mu}
\end{align}
\end{small}
Therefore, substituting (\ref{eqn:5}), (\ref{eqn:mu}) in the objective in (\ref{obj-dam}a) gives:
\begin{align}
%& \frac{1}{\sum_{m=1}^{M}{\lambda_m}}\sum_{m=1}^{M}{q_m \sum_{k=1}^{n}{\frac{a_k k!}{\mu_m^{k-1}}}} \nonumber \\
%=\; & 
\frac{1}{\sum_{m\in M}{\lambda_m}}\sum_{m\in M}{q_m \sum_{k=1}^{n}{a_k k! \left(\frac{q_m/\lambda_m}{1-q_m}\right)^k}}
\label{eqn:objective}.
\end{align}
\noindent
Further, the constraint in (\ref{obj-dam}b) can be simplified as:
\begin{align}
\sum_{m\in M}{{\lambda_m \delta(m) \left(\frac{\mu_m+\lambda_m-\lambda_m}{\mu_m+\lambda_m}\right)}} 
&= \sum_{m\in M}{\lambda_m \delta(m) (1-q_m)}. \nonumber 
%& \leq  \Delta \sum_{m=1}^{M}{\lambda_m}.
\label{eqn:constraint}
\end{align}

Now we present the final formulation.
\subsubsection*{$\mathrm{\mathbf{DARE-\Delta \; Retention\; Formulation}}$}
\begin{subequations}
\label{formulation:damage}
\begin{align}
\minimize_{q_m\in \mathbf{q}} & \frac{1}{\sum_{m\in M}{\lambda_m}}\sum_{m\in M}{ \sum_{k=1}^{n}{a_k k! \frac{q_m^{k+1}}{\lambda_m^k(1-q_m)^k}}}\\
\text {subject to:} \;  & \frac{1}{\sum_{m\in M}{\lambda_m}}\sum_{m\in M}{\lambda_m \delta(m) (1-q_m)}\leq \Delta
\\ 
&0 \leq q_m\leq 1, \forall \; m\in M
\end{align}
\end{subequations}
%The next Lemma proves that the above objective function is convex.
\begin{lemma}
\label{lem:convex-objective}
The objective function in the damage formulation (\ref{formulation:damage}) is convex.
\end{lemma}
\begin{proof}
See Appendix A.
\end{proof}
The constraint in formulation (\ref{formulation:damage}) poses upper and lower bounds on the value of $q_m={\lambda_m}/{(\mu_m+\lambda_m)}$. 
The boundary cases are: when $q_m= 0$ then $\mu_m=\infty$ which means that files are never written into the cache; alternatively, $q_m=1$ implies $\mu_m=0$ meaning that the file is retained forever. 
Once we obtain optimal $q_m$'s, the optimal $\mu's$ can be obtained by letting $\mu_m={\lambda_m(1-q_m)}/{q_m}$. The objective function in the optimization problem in (\ref{formulation:damage}) is convex (see Lemma~\ref{lem:convex-objective}).
%and can be solved using the interior point algorithm. 
We use a MATLAB convex program solver to solve (\ref{formulation:damage}) and report the results in Figure \ref{fig:difffunc}. We now summarize our numerical results.

\subsection{Damage-delay trade-offs for various damage functions with DARE-$\Delta$}
\par We study the delay-damage trade-offs obtained for three polynomial\footnote{The problem of finding suitable coefficients for the polynomial damage function could be an independent research problem by itself (left as an open problem for device engineers \cite{liu2012optimizing}) and is thus not considered in this work. Our work is concerned with finding optimal caching policies \emph{given} any polynomial damage function.} damage functions (linear, quadratic and cubic) on Poisson arrivals modulated with ZipF popularities ($\lambda_m=1/m^\alpha$, $\alpha=0.85$). We assume a unit delay for fetching files ($\delta(m)=1, m\in M$) for exposition. Note that with a unit delay we have $\Delta=\epsilon$, where $\epsilon\in[0,1]$ denotes the expected fraction of cache misses. We study the damage function trade-off with increasing $\epsilon$ for an increasing number of files $|M|$, as shown in Figure \ref{fig:difffunc}.

\par We observe that damage decreases with increasing $\epsilon$ in each case. This is reasonable since a higher $\epsilon$ means a relaxed delay constraint which implies that now more files can be written with lower retention values thus incurring less damage. We also observe that the value of damage increases with increasing number of files for the same value of $\epsilon$, which is expected as now more files are written in the cache (causing a higher damage) to achieve the required $\epsilon$. 

\begin{figure}
  \centering
  \includegraphics[scale=0.30]{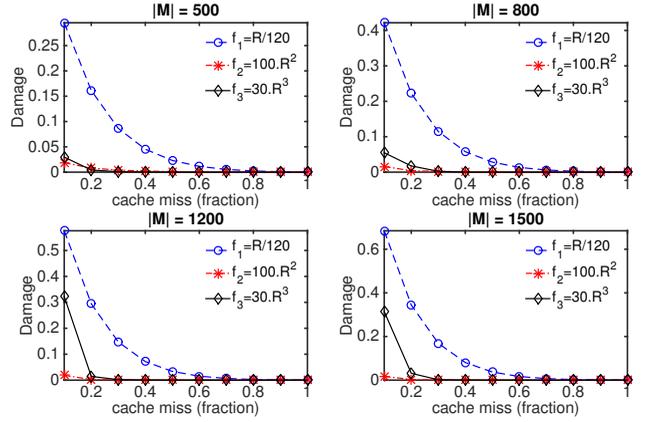}
  \caption{The figure shows the objective function values for Poisson arrivals modulated with ZipF $\alpha=0.85$ when plotted against increasing (allowed) fraction of cache misses, $\epsilon$, for a unit delay.} 
 \label{fig:difffunc}
\end{figure}

\section{Flash-aware Optimal Online Caching for A Finite Capacity Cache}
\label{sec:fincache}
In this section, we use the same model as defined in Section \ref{subsec: on_model} with the only difference that now the cache is finite and can contain only $B$ files. Upon a cache miss, a file is written in the cache for a duration given by the optimal retention time obtained in Section {\ref{sec:optonpol}}. A cache miss can result in a file eviction if the cache is full. We aim to obtain the optimal file to evict on every cache miss when the cache is full using only the knowledge of the past requests and cache contents. We formulate the problem of finding an optimal eviction sequence as a sequential decision problem using the theory of Markov Decision Processes (MDP). We then characterize the optimal solution which results in a very simple, easy to implement rule. We conclude the section by giving an outline of the DARE-$\Delta$ policy and comparing its performance with LRU, FIFO, RND policies.

\par Our work is a significant generalization of \cite{omri} where the authors have proposed a stationary, Markovian policy to optimally evict a file when files have non-uniform costs and the cache is finite. In contrast with \cite{omri} where the files are evicted \emph{only} upon a cache miss when the cache is full, in our model files leave the cache not only because they are evicted but \emph{also} because their retention time has \emph{expired}. Although subtle, this difference is significant as the minimization is performed over different file sets in both the cases. Hence the optimal solution in \cite{omri} is not a solution to our problem and vice versa. Moreover, modeling retention time for every file makes the analysis significantly more involved.
%\par We now explain the specifics of the MDP formulation.

\subsection{Markov Decision Process}
%In other words, we find an optimal stationary eviction sequence using Markov Decision Process by using the $\mu^*$s' found before.
\subsubsection{State Description}
\par We construct an MDP on a continuous time, discrete state space and use uniformization \cite{omri} to obtain a discrete time Markov chain (DTMC) from the continuous Markov process. Let $t=1,2,\dots T$ denote the time indices corresponding to the state transitions marked by file arrivals and file departures. Let $\mathbf{S}(t)$ be a state in the Markov Chain denoted by a 3-tuple, $\mathbf{S}(t)=\{S(t),R(t),D(t)\}$, where $S(t)$ is the set of files in the cache at $t$, $R(t)$ denotes the file requested at time $t$ and $D(t)$ is the first file departing at time $t$. We assume that a transition is either due to a file arrival or a file departure and \emph{not} both. For a file arrival, $D(t):=0$ and for a file departure, $R(t):=0$. Thus, the states of the MDP are of the form $\{S(t),R(t),0\}$ or $\{S(t),0,D(t)\}$. Files leave the cache either because they are evicted or because their retention time expires. A file whose retention time expires is said to \emph{depart} from the cache.

\par The cache state transitions can be summarized as follows. When file $D(t)$ departs from the cache $S(t)$, the cache becomes $S(t)-D(t)$. If there is a file arrival which results in a cache hit (i.e. $R(t)\in S(t)$) then the cache content at time $t+1$ is the same as that at time $n$ (i.e. $S(t+1)=S(t)$). In the case of a cache miss, two cases arise: (1) if the cache is not full then the new file gets added to the cache, i.e.  $S(t+1)=S(t)+R(t)$; (2) If the cache is full, then, the state at time $t+1$ is $S(t)+R(t)-U(t)$ where $U(t), U(t)\in S(t)+R(t)$ is the random variable denoting the file evicted on $n^{th}$ arrival on a full cache. Note that we assume \emph{optional} evictions, i.e. the policy may not evict a stored file upon a cache miss (in which case we say that the requested file $R(t)$ itself is instantaneously evicted). 
Formally,

$S(\text{$t+1$}) =T(\mathbf{S}(t), U(t))\\
=\begin{cases} S(t) &\mbox{if } R(t)\in S(t), |S(t)|\leq B \\ 
S(t)+R(t) & \mbox{if } R(t)\notin S(t), |S(t)|< B\\
S(t)+R(t)-U(t) & \mbox{if } R(t)\notin S(t), |S(t)|= B\\
S(t)-D(t) &\mbox{if } R(t)=0, |S(t)|\geq 1
\end{cases}$\\

Our goal is to find the optimal eviction sequence $U(t)$, $t=1,2,\dots, T,$ using MDP by using the optimal values of $D(t)$ (i.e.  the retention times $\sim exp(\mu_j)$, $j\in M$) obtained in Section \ref{sec:optonpol}.

\subsubsection{Markovian Policy}
\par It is easy to see the state $\mathbf{S}(t+1)$ only depends on state $\mathbf{S}(t)$ and $U(t)$. Thus, we need to focus only on Markovian policies (deterministic or randomized) that give optimal eviction sequences. 
Let $\mathcal{P}$ denote the set of all Markovian policies for evicting files. A policy $\pi\in \mathcal{P} $ is of the form $\pi=\{\pi_1, \pi_2,\dots\ \pi_T\}$, where each $\pi_t$ is a mapping from state $\mathbf{S}(t)$ to the evicted file in $\{0,1,\dots, M\}$, \emph{i.e.} $U(t)=\pi_t(\mathbf{S}(t))$. We define $U(t):=0$ when: (1) no eviction decision needs to be made (i.e. $R(t)\in S(t)$); (2) there is a cache miss and $U(t)$ refers to a file not present in cache or request (i.e. $R(t)\notin S(t)$ and $U(t)\notin S(t)+R(t)$). Let $\pi_t(u, \mathbf{S}(t))$ be the probability that policy $\pi$ evicts file $u$ in state $\mathbf{S}(t)$ on $n^{th}$ arrival, where $u \in \mathcal{M}$, then  $\pi_t(u, \mathbf{S}(t))$ satisfies the following properties: 
\begin{subequations}
\label{policy}
\begin{align}
&\sum_{u\in M}{\pi_t(u, \mathbf{S}(t))}=1, \nonumber \\
&\pi_t(u, \mathbf{S}(t))=0\;\; \forall u>0 \text{ if } R(t)\in S(t), \nonumber\\
&\pi_t(u, \mathbf{S}(t))=0\;\; \forall u: u \notin S(t)+R(t), R(t)\notin S(t), \nonumber \; \;
\end{align}
\end{subequations}
\subsubsection{State transition probabilities}
For our DTMC with state transitions due to file request arrivals and file departures, the probability of leaving a state due to an arrival of file $r$ is given by $\hat{p}_r=\lambda_r/{\sum_{m\in M}{(\lambda_m+\mu_m)}}$ and due to a departure of file $d$ is $\tilde{p}_d=\mu_d/{\sum_{m\in M}{(\lambda_m+\mu_m)}}$ since files have exponential interarrivals and retentions (as defined in Section \ref{subsec: on_model}). Let $\mathbf{p}$ denote the pmf of these probabilities. Let $\mathbf{P}_
{\pi}, \mathbf{E}_{\pi}$ denote the probability measure, expectation (respectively) under pmf $\mathbf{p}$ and policy $\pi$ and let $\mathbf{1}[\cdot]$ be the indicator function then we derive the state transition probabilities as follows:
\begin{align}
& \mathbf{P}_{\pi}[U(t)=u|\mathbf{S}(t)]=\pi_t(u, \mathbf{S}(t)), u\in M  \label{tran-1} \\
&\mathbf{P}_{\pi}[S(t+1)=\tilde{S}, R(t+1)=r, D(t+1)=0|\mathbf{S}(t), U(t)] \nonumber \\
 &= \hat{p}_r\times  \mathbf{P}_{\pi}[S(t+1)=\tilde{S} |\mathbf{S}(t), U(t)] \nonumber \\
 &=\hat{p}_r\times \mathbf{1}[T(\mathbf{S}(t),U(t))=\tilde{S}] \label{tran}\\
&\mathbf{P}_{\pi}[S(t+1)=\tilde{S}, R(t+1)=0, D(t+1)=d|\mathbf{S}(t)] \nonumber \\
 &= \tilde{p}_d \times  \mathbf{P}_{\pi}[S(t+1)=\tilde{S} |\mathbf{S}(t)] \label{tran-2}
\end{align} 
\par Equations (\ref{tran-1})-(\ref{tran}) follow since IRM file arrivals are independent of the state of the cache and the time of the request.  
Equations (\ref{tran})-(\ref{tran-2}) apply for every $(\tilde{S}, r, d)\in \mathbf{S(t+1)}$. 

\subsubsection{Cost function}
\par A one-shot cost $c(m)$ for file $m$ (as in Section \ref{sec:optonpol}) is incurred on every cache miss. The expected cost for the horizon of length $T$ under the policy $\pi$ becomes:
\begin{align}
J_c(\pi,T)=\mathbf{E}_{\pi}\left[\sum_{t=0}^{T}\mathbf{1}_{[R(t)\notin S(t)]}\times c(R(t))\right] \nonumber 
\end{align}

\par The average cost over the horizon of $T$ discrete time steps under policy $\pi$ is given by,
$J_c(\pi)=\limsup_{T \rightarrow \infty}\frac{{\sum_{m\in M}{(\lambda_m+\mu_m)}}}{T+1}J_c(\pi,T)$\footnote{It is possible that with an arbitrary policy $\pi$ the limit may not exist, therefore we use supremum which is a standard practice in the MDP literature.}.

\subsection{The Optimal Eviction Policy}
Now we formulate and solve the MDP to find an optimal eviction policy. We define $J_c(\pi,(S,R,D),T)$ as the cost-to-go for the policy $\pi$ starting in the state $\mathbf{S}=\{S, R, D\}$. Minimizing $J_c(\pi,(S,R,D),T)$ at every possible state will give us the optimal eviction policy.
\begin{align}
& J_c(\pi,(S,R,D),T):= \nonumber\\
&\mathbf{E}_{\pi}\left[\sum_{t=0}^{T}\mathbf{1}_{[R(t)\notin S(t)]}\times c(R(t)|\mathbf{S}(0)=\{S, R, D\})\right]. \nonumber
\end{align}
We will use the \emph{value iteration} approach to solve our problem. The value function minimizes cost-to-go over all policies, i.e. $V_T(S,R,D)=\inf_{\pi\in \mathcal{P}}{J_c{(\pi,(S,R,D),T)}}$.

\par Next, we write the Dynamic Programming Equation (Bellman equation) for this MDP. We form two different recurrence equations for the states of type $(S,r,0)$ and $(S,0,d)$, each accounting for a file request and a departure (recall that no other types of states are possible as we have assumed that file requests and departures are mutually exclusive). We first state the recurrence equations followed by the explanation:
\small
\begin{align}
\label{dpe}
\nonumber
V_{T+1}(S,r,0)&=\mathbf{1}_{\{r\in S\}}\mathbb{E}_{R^*}[V_T(S,R^*,0)] \\  \nonumber 
&+ \mathbf{1}_{\{r\notin S,|S|<B\}}\left(c(r)+\mathbb{E}_{R^*}[V_T(S+r,R^*,0)]\right) \\ \nonumber
&+ \mathbf{1}_{\{r\notin S,|S|=B\}}  \nonumber\\
&\left(c(r)+\min_{u\in S+r}\mathbb{E}_{R^*}[V_T(S+r-u,R^*,0)]\right), \nonumber \\
&+ \mathbf{1}_{\{r\in S\}}\mathbb{E}_{D^*}[V_T(S,0,D^*)]\nonumber \\  
&+ \mathbf{1}_{\{r\notin S,|S|<B\}}\left(c(r)+\mathbb{E}_{D^*}[V_T(S+r,0,D^*)]\right) \nonumber\\ 
&+ \mathbf{1}_{\{r\notin S,|S|=B\}}\nonumber\\
&\left(c(r)+\min_{u\in S+r}\mathbb{E}_{D^*}[V_T(S+r-u,0,D^*)]\right)\\
%\end{align}
%% % % % % % % % % % % % % % % % % % % % % %
%\begin{align}
V_{T+1}(S,0,d)&=\mathbb{E}_{R^*}(V_T(S-d,R^*,0)) \nonumber\\
&\; +\mathbb{E}_{D^*}(V_T(S-d,0,D^*))
\label{dpe2}
\end{align}
\normalsize
The different terms in (\ref{dpe})-(\ref{dpe2}) can be explained as follows:
\begin{itemize}
\item $V_{T+1}(S,r,0)$ is the value of the objective when optimal action is taken in the state $S,r,0$ at time $t=0$ to minimize the cost over the horizon $[0,T+1]$.  
The first (or fourth) term in the sum says that when file request $r$ belongs to the cache $S$ then the expected cost for horizon $[0,T]$ due to a file request $R^*$ (or a departure $D^*$) at time $t=0$ is given by  $\mathbb{E}_{R^*}[V_T(S,R^*,0)]$ (or $\mathbb{E}_{D^*}[V_T(S,0,D^*)])$.
\item The second and fifth terms differ from the above in that the file $r$ requested at $t=0$ leads to a cache miss but the cache is not full so the requested file is written in the cache (without any eviction) thus increasing the expected cost over the horizon $[0,T]$ by $c(r)$.
\item The third and sixth terms capture the case when the cache is full at $t=0$ and there is a cache miss upon request thus leading to a file eviction. The expected cost over the horizon $[0,T]$ is thus obtained by minimizing over all possible evictions, i.e. $u\in S+r$.
\end{itemize}
\par Equation \ref{dpe2} represents file $d$ departing from the cache at time $t=0$. Here no cost is incurred over the horizon $[0,T+1]$ since we do not fetch or write a new file. The two terms in the sum infer that the next state could be due to a file request or a file departure at $t=1$.

\begin{figure}
\centering
\includegraphics[scale=0.26]{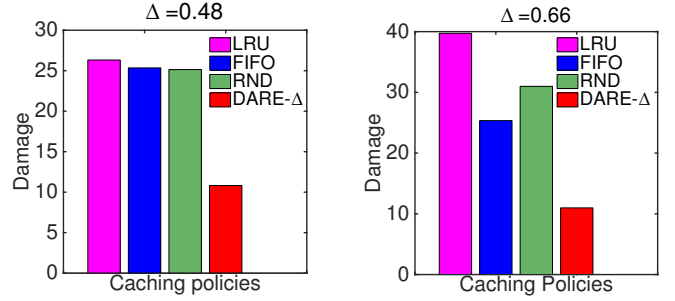}
\caption{Damage values for LRU, FIFO, RND and MDP with Poisson arrivals: (1) ZipF-0.65 popularity,  $\Delta=\epsilon=0.48$, (2) ZipF-0.95 popularity, $\Delta=\epsilon=0.66$.}
\label{fig:D_1}
\end{figure}
  
\par By inspection, we observe that in Bellman equations~(\ref{dpe})-(\ref{dpe2}), we \emph{only} need to optimize the term: 
$\min_{u\in S+r}\mathbb{E}_{R^*}[V_T(S+r-u,R^*,0)] \label{rec}$\footnote{There are two minimization terms in Bellman equation, however, the second term can be minimized only by minimizing the first term due to the recurrence relation. See details in the Proof of Theorem~\ref{main-theorem} in Appendix.}. Theorem~\ref{main-theorem} characterizes the optimal eviction policy obtained from the minimization.

\begin{theorem}
To minimize the expected cost over the horizon $[0,T]$, the optimal eviction policy evicts a file $v$ in the state $(S,r,0)$ that satisfies the following:
\begin{align}
%&\arg\min \{u\in S+r:\mathbb{E}[V_T(S+r-u,R^*,0)]\}\nonumber\\
v=\arg\min_{u\in S+r} \{p_uc(u)\}
\end{align}
whenever the cache is full and there is a cache miss on the request $r$ (i.e. $r\notin S$). 
\label{main-theorem}
\end{theorem}

\begin{proof}
See Appendix C.
\end{proof}
\par Theorem \ref{main-theorem} characterizes a \emph{stationary} optimal  Markov eviction policy which suggests evicting a file that is requested least often and can be fetched, written with the least cost. Intuitively, the eviction rule seems fairly reasonable. We note that the cost function $c(u)$ is general as it can be easily extended to represent an convex combination of various factors such as file size, avaliable bandwidth on the channel (from where file $u$ is fetched), etc., apart from the fetching and writing cost. 

\subsection{DARE-$\Delta$ end-to-end policy design}
\label{sec:complete_policy}
Next we outline the complete description of DARE-$\Delta$. DARE-$\Delta$ executes the following routine:

\par {\textbf{(I) Preprocessing}}: Given the set
of files ($M$), an acceptable delay ($\Delta$), obtain the optimal retention time 
parameters of all files (i.e. $\mu_i, i\in M$) by solving the convex optimization problem in (\ref{formulation:damage}). Further, sample retention times $\mathscr{R}_i,i\in M$ where each retention time, $\mathscr{R}_i$, is an exponential random variable sampled from mean $\mu_i$.

{\textbf{(II) Run-time execution}}: Given a request for file $r\in M$ at time $t\in \{1,2,\dots,T\}$, check if the cache $S$ contains the requested file. If so, serve $r$ instantaneously. If not, then find the file $k$ in cache such that $k=\arg\min_{u\in S+r} p_uc(u)$ (see Theorem \ref{main-theorem}). If file $k=r$, i.e. file $k$ is the request itself then do nothing. If file $k$ is a file from the cache then two cases arise: (a) if the cache is not full then write the requested file $r$ in cache for a retention duration $\mathscr{R}_r$; (b) if the cache is full then write the requested file $r$ with retention $\mathscr{R}_r$ and evict file $k$. 

\subsection{Numerical comparison against other online policies}
Recall that the caching policies LRU, FIFO and RND \emph{assume} that the files are written in the cache until evicted. For comparison, we embed the notion of retention time in the well-known policies by first assuming that the files are written in the cache for a deterministic time thus incurring a one-shot damage on each write. Second, we even optimize these policies by finding the \emph{best} such time for each policy. We do this by simulating the policies over a wide range of time values and finding a time that yields minimum damage if all files are written in cache for that time. 

We consider Poisson arrivals modulated with ZipF-$\alpha$ with $\alpha \in\{0.65, 0.95\}$, respectively. We consider unit delay and fix a value for expected cache misses, $\Delta=\epsilon$, in (\ref{formulation:damage}). Further, we simulate DARE-$\Delta$ by writing files in the cache for the optimal retention time computed from the solution of (\ref{formulation:damage}) with cost $c(m)=f(m)$ for the chosen value of $\epsilon$. The damage function for writing files is assumed to be quadratic in retention time. Upon a cache miss, DARE-$\Delta$ evicts the file $u$ with least $p_uc(u)$ (see Theorem $\ref{main-theorem}$). 
\par The results in Figure \ref{fig:D_1} show that \emph{even after optimizing} the existing caching policies over all possible retention times, DARE-$\Delta$ outperforms other policies by giving a \emph{2-4 fold} damage savings, thus agreeing with the analytical result (derived in Theorem \ref{main-theorem}). Moreover, we note that (1) FIFO and RND differ significantly with respect to damage under IRM, but are known to perform similar in terms of cache miss. (2) LRU and RND, which are being actively considered for deploying in CCN caches, perform very poor with respect to damage.
% %

%%
\section{Bridging the gap between offline and online policies}
\label{sec:policyperformance}
\begin{figure*}[tbp]
\centering
\includegraphics[scale=0.41]{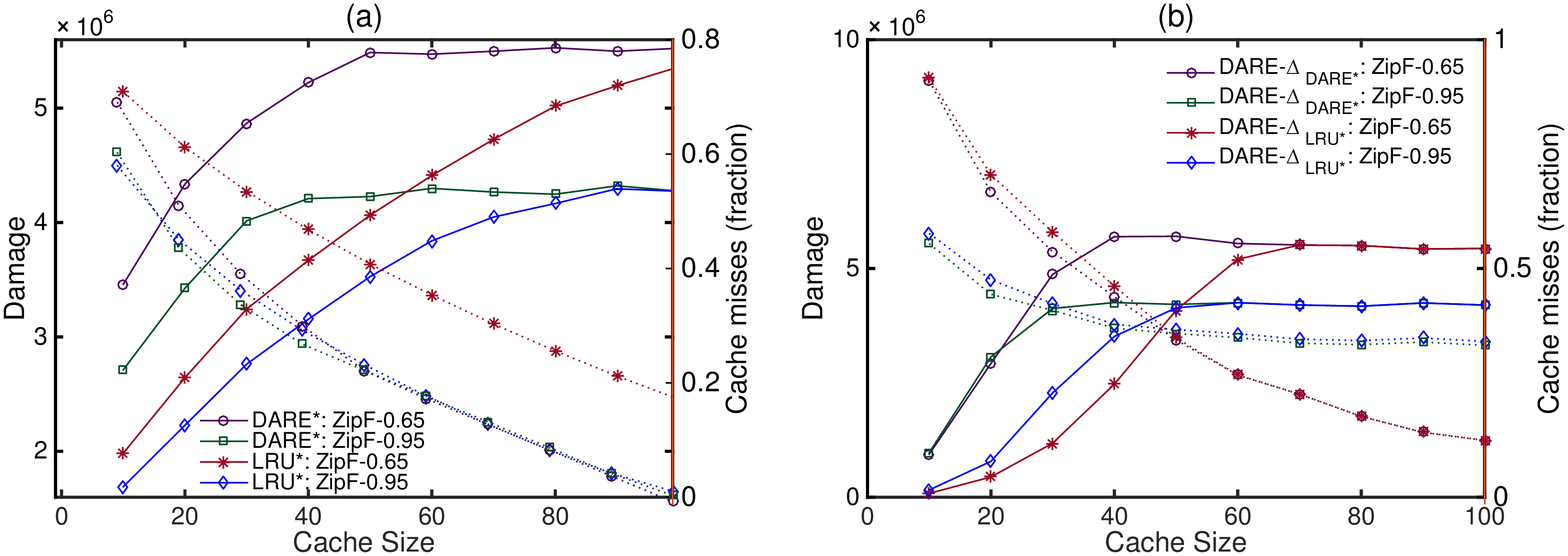}
\caption{Delay-damage trade-offs with increasing cache size and varying ZipF parameters for (a) Offline policies: DARE* and LRU*, (b) Online policies: DARE-$\Delta_{\text{DARE*}}$ and DARE-$\Delta_{\text{LRU*}}$. The bold lines show damage, the dotted lines show cache miss (fractions).}
\label{fig:D_5.1}
\end{figure*}

\par So far we described DARE-$\Delta$ and have shown some numerical results on its performance. In this section, we set up a simulation framework to test the large cache approximation. We benchmark the performance of DARE-$\Delta$ against LRU and FiF policies. Recall that the offline policies in Section \ref{sec:optoffpol} assume that all cache misses are allocated, i.e. a requested file \emph{must} be written to the cache in the event of a cache miss (see Section \ref{sec:optoffpol}). However, from a device damage perspective, allocating every cache miss is not necessarily optimal. For example, if there is a cache miss on a request for a very unpopular file then it can be served directly by fetching it from the server instead of writing it in the cache at the expense of evicting a more popular file. This practice of selecting when to cache a file and when not to is particularly useful in mitigating expensive write damage in a flash memory \cite{Pritchett}.

\par We thus obtain the modified variants of DARE, FiF and LRU, referred henceforth as DARE*, FiF* and LRU*, by allocating cache misses. That is, we now allow the policy to not cache a requested file \emph{if} it is going to be requested farthest in the future (for FiF* and DARE*) or is the least recently used (for LRU*). Note that, our online policy already allocates cache misses since it has the option to evict the request itself. 

%\subsection{Simulation set-up}
\par The performance analysis is based on the following parameters. We are interested in time asymptotics so we first generate a long request string corresponding to a horizon of length $T=10^5$ slots (as in the offline case) or transitions (as in the online case). The generated file requests are sampled from $M=200$ files (of equal size). File requests form a Poisson process modulated with ZipF-ian popularity as before, i.e. the probability of requesting file $i$ is proportional to $1/i^\alpha$, with the sum of request probabilities normalized to one. We consider $\alpha=\{0.65,0.95\}$. We obtain the optimal retention times from Section \ref{sec:optonpol} and the optimal file to be evicted on each miss from Section \ref{sec:fincache}. The damage function for writing files is assumed to be quadratic in retention time. Cache size $B$ is varied from 10 files to 100 files in steps of 10. For each value of cache size, we obtain results and average it over 1000 iterations. Damage (or delay) for a particular cache size is calculated by obtaining the average damage (or fraction of cache misses) over all iterations.

%\subsection{Results}
\subsection{Damage-delay trade-offs}
\par We evaluate the damage-delay trade-off with increasing cache sizes for two settings. We first show the delay-damage trade-off for DARE* versus LRU* with increasing cache size in Figure \ref {fig:D_5.1}(a). The solid lines indicate the damage curve and the dotted lines indicate the delay curve. We observe that as the cache size increases, the damage increases and delay decreases, which is consistent with our observations in Section \ref{sec:optoffpol}. Moreover, we note that a higher value of $\alpha$ results in a lower damage. This is expected since popular files get a larger share with increasing $\alpha$ (i.e. the disparity between a popular versus a non-popular file increases) thus sufficing to store a few most popular files.

\begin{figure}
\centering
\includegraphics[scale=0.28]{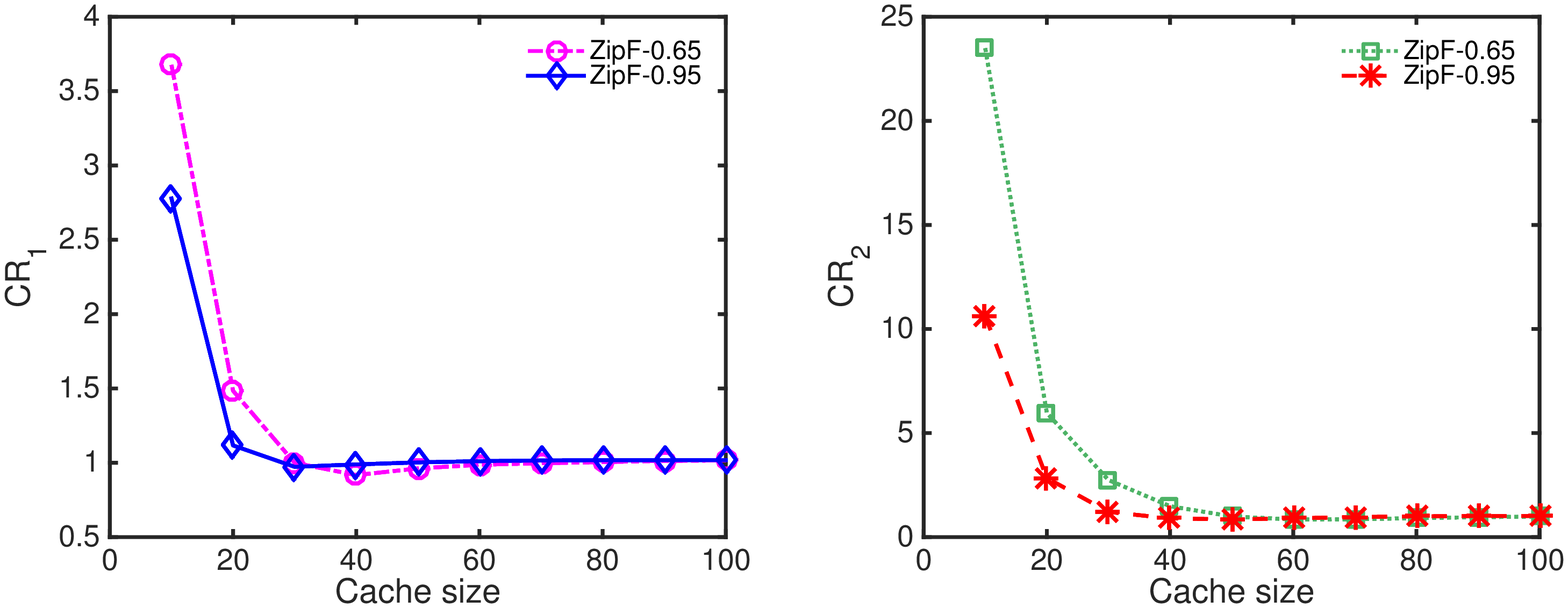}
\caption{Competitive ratios for the online policies with varying cache size and ZipF exponents.} 
\label{fig:D_5.2}
\end{figure}

We further use the delay (i.e. cache miss fraction, $\epsilon$) obtained from DARE* and LRU* for each cache size and provide it as an input to the online policy DARE-$\Delta$. The results are plotted as DARE-$\Delta_{\text{DARE*}}$ and DARE-$\Delta_{{LRU*}}$, respectively, in Figure \ref {fig:D_5.1} (b). Similar trends as above were observed in damage-delay trade-offs. Moreover, we note that the damage converges to a steady state with increasing cache size. Also, even though the retention times were obtained by feeding $\epsilon$ from DARE* and LRU* (see the dotted lines in Figure \ref {fig:D_5.1} (a)), the resultant delay obtained from variants of DARE-$\Delta$ was found to be moderately higher than the original $\epsilon$. This is the price of uncertainty paid when shifting from the offline caching, which has a complete knowledge of request arrivals, to the online caching, which only knows the value of the expected delay (cache miss fraction, $\epsilon$). 

%\par Next we compare damage of offline and online policies.

\subsection{Competitive ratio}
\par A common tool to estimate the quality of an online algorithm (i.e. DARE-$\Delta)$ against an offline algorithm (i.e. DARE*) is to derive bounds on the ratio of the offline cost to the online cost given the worst case file requests; this ratio is called competitive ratio (CR). In our analysis, file requests conform to Poisson distribution thus limiting the possibility of a pathological worst-case input. We obtain CR for a long request string of length $T= 10^5$ (with other parameters same as above) to calculate two quantities.
(1) $CR_1=\frac{\text{damage from DARE*}}{\text{damage from DARE-}\Delta_{\text{DARE*}}}$, (2) $CR_2=\frac{\text{ damage from LRU*}}{\text{damage from DARE-}\Delta_{\text{LRU*}}}$. A value $CR=r$ shows an r-fold superiority of the online algorithm over the offline counterpart. The plot in Figure \ref{fig:D_5.2} shows the results for $CR_1$ and $CR_2$. 
\par We observe that the online cost is always lower than the offline cost resulting in both $CR$s' to take a value greater than one. Moreover, we observe that both the $CR$s' start with a higher value and gradually converge to one. This shows that our optimal online policy converges very fast to the damage performance of the optimal offline policy with increasing cache size. We briefly justify these observations as follows. Recall that the online policy obtains retention times for an infinte capacity cache whereas the offline policy uses a finite capacity cache. The only interaction between offline and online policies is via the delay (or the cache miss fraction, $\epsilon$) which we obtain from the offline policy and pass as a parameter to the online policy. Thus, for smaller caches, due to the cache capacity constraint, the offline policy incurs a higher cost compared to the online policy which assumes an uncapacitated cache while calculating file retentions. Nevertheless, as the cache size grows, the discrepancy between the two policies vanish and the cost incurred by both offline and online policies match up. 
% %

\section{Conclusions}
\label{sec:conclusions}
\par This paper advances the state-of-the-art of traditional data caching literature when applied to CCN caches by proposing a cross-layer optimization for the network layer objective of minimizing the content retrieval delay and the device layer objective of minimizing the flash damage. 
%reformulating the problem taking the flash memory damage into account. 
We analyze the delay-damage trade-offs for both offline and online caching to obtain optimal damage-aware caching policies. Our results demonstrate that our policies achieve significant damage reductions when compared to the traditional caching policies with the same delay bounds.
%We analyze the cache-miss versus flash-damage trade-off for both the offline and online caching; 
%we obtained optimal damage-aware caching policies and compared their performance against well-known caching policies for a variety of parameters. 
This advocates using damage-aware caching policies in data intensive applications where flash memory cost and wear-out are of critical importance. 
\par One possible direction for future work is to consider temporal correlations in file request arrivals. Another direction is to extend the problem to a network of caches. Finally, it is an open problem to devise a framework that jointly optimizes over both retention time durations (i.e. all possible distributions) and eviction sequences.
%%%

\bibliographystyle{IEEEtran}
\bibliography{phd}

\section*{appendices}
\label{sec:appendixA}

\section*{Appendix A: Proof of Lemma \ref{lem:convex-objective}}
We want to show that the function, $h(\mathbf{q},\boldsymbol{\lambda},M):= \frac{1}{\sum_{m=1}^{M}{\lambda_m}}\sum_{m=1}^{M}{ \sum_{k=1}^{n}{a_k k! \frac{q_m^{k+1}}{\lambda_m^k(1-q_m)^k}}}$ is convex. Let $\sigma(k,m,M):=\frac{a_kk!}{\lambda_m^k \sum_{m=1}^M{\lambda_m}}$. Note that $\sigma(k,m,M)$ does not depend on $q_m$. We prove the claim by first showing it for $M=1$, where we have $h(q_1,\lambda_1,1)= { \sum_{k=1}^{n}{\sigma(k,1,1) \frac{q_1^{k+1}}{(1-q_1)^k}}}$. Let $H(q_1,\lambda_1,1)$ be the double derivative of $h(q_1,\lambda_1,1)$. For convexity, we require, 
$H(q_1,\lambda_1,1)= \frac {\partial^2{h}}{\partial{q_1^2}}\geq0.$
Now, the first derivative with respect to $q_1$ is:
\begin{align}
\frac{\partial h}{\partial q_1}={ \sum_{k=1}^{n}{\sigma(k,1,1) \frac{q_1^{k}(k+1-q_1)}{(1-q_1)^{k+1}}}} \nonumber
\end{align}
and, the second derivative, after simplifying, is:
\begin{align}
\frac{\partial^2 h}{\partial q_1^2}={ \sum_{k=1}^{n}{\sigma(k,1,1) \frac{q_1^{k-1}(1-q_1)^k k(k+1)}{(1-q_1)^{2(k+1)}}}} \nonumber
\end{align}
which is well defined and non-negative as every term in the expression is non-negative provided $q_1\in [0,1)$.\\
Thus the objective in (\ref{formulation:damage}) is convex because it is the sum of $M$ different convex functions, thus proving the claim.

\section*{APPENDIX B: Proof of Theorem \ref{thm:opt_off}}
The optimality theorem follows from Lemmas \ref{belady-opt} and \ref{arbit-ret}.
%\begin{proof}
\begin{lemma} \label{belady-opt}
For the eviction sequence given by FiF policy, DARE policy gives optimal offline cost.
\end{lemma} 
\begin{proof}
DARE incurs the same delay cost as FiF (since both incur optimal cache misses). Thus to prove the lemma it suffices to prove the non-existence of a policy which assigns retention times to files as per the eviction sequence from FiF by incurring the same number of cache misses (delay cost) but a less retention cost than DARE. Let ${P^*}$ be that policy. This implies that there exists a file $l$ which is retained in cache for more slots with DARE compared to that with ${P^*}$. For $l$, there must exist a time triplet $(t_0,t_1,t_2)$ -- where $t_0<t_1<t_2$ -- such that for the interval $(t_0,t_1)$, $l$ is present in cache with both the policies; however, for the interval $(t_1,t_2)$, file $l$ is only present in cache with DARE but not with ${P^*}$, thus causing a less retention cost with $P^*$. DARE is designed to store $l$ in cache only for the time it is useful, indicating that $l$ is stored for the duration $(t_1,t_2)$ to account for a request for file $l$ at time $t_2$. However, at $t_2$, $P^*$ will result in a cache miss since it did not have $l$ cached, thus increasing the optimal number of cache misses by one. A contradiaction. This proves the claim.

\end{proof}

\par Let $\mathcal{F}$ denote the family  of eviction sequences with optimal ($F$) misses. While DARE is built on eviction sequence from FiF (denoted $E$), we now prove that DARE is optimal over all $J \in \mathcal{F}$.

\begin{lemma}\label{arbit-ret}
The retention cost incurred with $J\in \mathcal{F}, J\ne E$ is greater or equal to the retention cost incurred by $E$.
\end{lemma}

\begin{proof}
We provide a brief proof sketch due to space constraints. Suppose there exists an optimal sequence $J^*\in \mathcal{F}$, $J^*\neq E$, such that the total retention cost with $J^*$ is less than that with $E$. We transform every eviction in $J^*$ to the evictions in $E$ using the \emph{exchange argument} and show that the one-shot retention cost of a file in $E$ is a permutation of the one-shot retention cost of a file in $J^*$, thus making the cumulative retention costs equal which is a contradiction.
\end{proof}
This completes the proof of Theorem 1.

\section*{APPENDIX C: Proof of Theorem \ref{main-theorem}}
\label{app:F}
There are two minimization terms in Bellman equations (\ref{dpe})-(\ref{dpe2}). We use induction for minimizing the first term in (\ref{dpe}), i.e., $\min_{u\in S+r}\mathbb{E}_{R^*}[V_T(S+r-u,R^*,0)]$. We then show that for the second term, i.e. $\min_{u\in S+r}\mathbb{E}_{D^*}[V_T(S+r-u,0,D^*)]$, the proof proceeds by induction similar to minimization of the first term and it reduces to minimizing the first term itself due to (\ref{dpe}), (\ref{dpe2}) after simplification.\\
% (see details in Appenidix C of \cite{dropboxlink}). \\
\textbf{Proving the result for the first minimization term:}\\
We define $v=\arg \min\{u\in S+r: p_uc(u)\}$, then, $\mathbb{E}[V_T(S+r-v,R^*,0)]=\min_{u\in S+r} \mathbb{E}[V_T(S+r-u,R^*,0)].$
The proof proceeds by induction on $T=0,1,\dots$. At each step, we want to show that \\
$\mathbb{E}[V_{T}(S+r-u,R^*,0)]-\mathbb{E}[V_{T}(S+r-v,R^*,0)]\geq 0 $\\
\textbf{The basis step}: Fix the initial state, request as $(S, r)$. Thus,\\
%\begin{align*}
$V_0(S,r,0)=1_{\{r\notin S, |S|<B\}}2c(r)+1_{\{r\notin S,|S|=B\}}2 c(r)$, and,
%\end{align*} 
\small
\begin{align*}
\mathbb{E}[V_0(S+r-u,R^*,0)]=& \mathbb{E}[(1_{R^*\notin\{S+r-u\},|S+r-u|<B}+\\& 1_{R^*\notin\{S+r-u\},|S+r-u|=B})2c(R^*)]\\
=& \mathbb{E}[1_{R^*\notin\{S+r-u\}}2c(R^*)]\\
=& \mathbb{E}[1_{R^*\notin\{S+r\}}2c(R^*)] +\mathbb{E}[1_{R^*=u}2c(R^*)]
\end{align*}
\normalsize
Now we write an expression for $\mathbb{E}[V_0(S+r-v,R^*,0)]$, and observe that,
\begin{align*}
&\mathbb{E}[V_0(S+r-u,R^*,0)]-\mathbb{E}[V_0(S+r-v,R^*,0)]\\
& =\mathbb{E}[1_{R^*\notin\{S+r\}}2c(R^*)] + \mathbb{E}[1_{R^*=u}2c(R^*)]\\& - \mathbb{E}[1_{R^*\notin\{S+r\}}2c(R^*)] - \mathbb{E}[1_{R^*=v}2c(R^*)]\\
& = \mathbb{E}[1_{R^*=u}2c(R^*)]- \mathbb{E}[1_{R^*=v}2c(R^*)]\\
& = 2(p_uc(u)-p_vc(v))
\end{align*}
That is the claim is true with $T=0$.\\
\textbf{The Induction step}:
Assume that the claim is true for some fixed $T>0$.
Fix $(S,r,0)$ and $(S,0,d)$ with $r\notin S$ and $d\in S$. We need to show that, for $u\in S+r$, we have, \\
%\begin{align}
$\mathbb{E}[V_{T+1}(S+r-u,R^*,0)]-\mathbb{E}[V_{T+1}(S+r-v,R^*,0)]\geq 0$\\
%\end{align}
To show this, we take the expectation of $V_{T+1}(S+r-u,R^*,0)$ and use (\ref{dpe}) to get,
\small
\begin{align}
&\mathbb{E}[V_{T+1}(S+r-u,R^*,0)] \nonumber \\
&= P[R^*\in S+r-u] E[V_T(S+r-u,R^{**},0))] \label{40} \\
&+ P[R^*\in S+r-u] E[V_T(S+r-u,0,D^{**}))] \label{41} \\
&+ \mathbb{E}[1_{R^*\notin\{S+r-u\},|S+r-u|<B}*2c(R^*)] \label{42}\\
&+ \mathbb{E}[1_{R^*\notin\{S+r-u\},|S+r-u|=B}*2c(R^*)] \label{43} \\
&+ \mathbb{E}[1_{R^*\notin\{S+r-u\},|S+r-u|<B}\mathbb{E}(V_T(S+r-u+R^*, R^{**},0))] \label{44} \\
&+ \mathbb{E}[1_{R^*\notin\{S+r-u\},|S+r-u|<B}\mathbb{E}(V_T(S+r-u+R^*,0, D^{**}))] \label{45}\\
&+ \mathbb{E}[1_{R^*\notin\{S+r-u\},|S+r-u|=B} \widehat{V}_T{(S+r-u,R^*,0)}] \label{46}\\
&+ \mathbb{E}[1_{R^*\notin\{S+r-u\},|S+r-u|=B} \widehat{V}_T{(S+r-u,0, R^*)}] \label{47}
\end{align}
\normalsize
where $\widehat{V}_T{(S+r-u,R^*,0)}:=\min_{u'\in S+r-u+R^*}\mathbb{E}[V_T(S+r-u+R^*-u',R^{**},0)]$ and\\ $\widehat{V}_T{(S+r-u,0,R^*)}:=\min_{u'\in S+r-u+R^*}\mathbb{E}[V_T(S+r-u+R^*-u',0,D^{**})]$.
Now, we state some \textit{reductions} (\ref{eq:key1}) to (\ref{eq:key4}) which will be instrumental in getting to the proof.
\begin{align}
& P[R^*\in S+r-u]\mathbb{E}[{V_T(S+r-u,R^{**},0)}]\nonumber \\
& = P[R^*\in S+r-(u,v)]\mathbb{E}[{V_T(S+r-u,R^{**},0)}] \nonumber \\
& +p_v\mathbb{E}[{V_T(S+r-u,R^{**},0)}]\label{eq:key1}\\
& \mathbb{E}[1_{R^*\notin \{S+r-u\}}c(R^*)]
=\mathbb{E}[1_{R^*\notin \{S+r\}}c(R^*)]+ p_uc(u) \label{eq:key2}
\end{align}
\begin{align}
& \mathbb{E}[1_{\{R^*\notin S+r-u\}}\widehat{V}_T(S+r-u, R^*,0)] \nonumber \\
& = \mathbb{E}[1_{R^*\notin \{S+r\}}\widehat{V}_T{(S+r-u,R^*,0)}]\nonumber\\
 & + p_u\widehat{V}_T(S+r-u,u,0) \label{eq:key3}\\
&\widehat{V}_T(S+r-u,u,0)=\min_{u'\in S+r}\mathbb{E}[V_T(S+r-u',R^{**},0)] \nonumber\\
&=\mathbb{E}[V_T(S+r-v,R^{**},0)] \label{eq:key4}
\end{align}
With the machinery to simplify the expressions, we write the difference in terms with $u$ and $v$ for (\ref{40}) through (\ref{47}). Recall reduction (\ref{eq:key1}) to express $\ref{40}(u)- \ref{40}(v)$ as:
\begin{align*}
& P[R^*\in S+r-(u,v)]\mathbb{E}[{V_T(S+r-u,R^{**},0)}]\\&+p_v\mathbb{E}[{V_T(S+r-u,R^{**},0)}]\\& -P[R^*\in S+r-(u,v)]\mathbb{E}[{V_T(S+r-v,R^{**},0)}]\\& -p_u\mathbb{E}[{V_T(S+r-v,R^{**},0)}]\\&
=P[R^*\in S+r-(u,v)]\\ & \times(\mathbb{E}[{V_T(S+r-u,R^{**},0)}]- \mathbb{E}[{V_T(S+r-v,R^{**},0)}])\\& +p_v\mathbb{E}[{V_T(S+r-u,R^{**},0)}]\\& - p_u\mathbb{E}[{V_T(S+r-v,R^{**},0)}]
\end{align*}
\par We know by induction on $T$ that, $\mathbb{E}[{V_T(S+r-u,R^{**},0)}]-\mathbb{E}[{V_T(S+r-v,R^{**},0)}]\geq 0$ holds. Thus, to prove $\ref{40}(u)-\ref{40}(v)\geq 0$, we need,\\
\begin{small}
%\begin{align}
$p_v\mathbb{E}[{V_T(S+r-u,R^{**},0)}]-p_u\mathbb{E}[{V_T(S+r-v,R^{**},0)}]\geq 0$
%\end{align}
\end{small}
\normalsize
\par Next, we invoke reduction (\ref{eq:key2}), define $p_S(B)=P(|S|<B)$ and simplify $\ref{42}(u)-\ref{42}(v)$ as follows:
\begin{align*}
& \mathbb{E}[1_{R^*\notin\{S+r-u\},|S+r-u|<B}2c(R^*)]\\&-\mathbb{E}[1_{R^*\notin\{S+r-v\},|S+r-v|<B}2c(R^*)]\\
= & \sum_{R^*\notin\{S+r-u\}}P[R^*||S|<B]P[|S|<B] 2c(R^*)\\
&-\sum_{R^*\notin\{S+r-v\}}P[R^*||S|<B]P[|S|<B] 2c(R^*)\\
= & \sum_{R^*\notin\{S+r-u\}} 2p_S(B)\times p_{R^*}c(R^*)\\& -\sum_{R^*\notin\{S+r-v\}} p_{R^*}p_S(B)2c(R^*)\\
=& \sum_{R^*\notin\{S+r\}} 2p_S(B)\times p_{R^*} c(R^*) + 2p_S(B)\times p_uc(u)\\
& - \sum_{R^*\notin\{S+r\}} p_{R^*}p_S(B)2c(R^*)-2p_S(B)\times p_vc(v)\\
=& 2p_S(B)(p_uc(u)-p_vc(v))
\end{align*}
which is $\geq 0$ because $p_uc(u)\geq p_vc(v)$ by definition of $v$.
Similarly, we show $\ref{43}(u)-\ref{43}(v)\geq 0$.\\
Now, we invoke reduction (\ref{eq:key3}) and consider $\ref{44}(u)-\ref{44}(v)$.
\small
\begin{align*}
&\mathbb{E}[1_{R^*\notin\{S+r-u\},|S+r-u|<B}\mathbb{E}(V_T(S+r-u+R^*, R^{**},0))]\\
& -\mathbb{E}[1_{R^*\notin\{S+r-v\},|S+r-v|<B}\mathbb{E}[V_T(S+r-v+R^*, R^{**},0))]\\
=& \sum_{R^*\notin\{S+r-u\}}P[R^*||S|<B]\times P[|S|<B]\\
& \times \mathbb{E}[V_T(S+r-u+R^*, R^{**},0)]\\
&-\sum_{R^*\notin\{S+r-v\}}P[R^*||S|<B]\times P[|S|<B]\\ & \times \mathbb{E}[V_T(S+r-v+R^*, R^{**},0)]\\
=& \sum_{R^*\notin\{S+r-u\}} p_{R^*}\times p_S(B)\times \mathbb{E}[V_T(S+r-u+R^*, R^{**},0)]\\
& -\sum_{R^*\notin\{S+r-v\}} p_{R^*}\times p_S(B)\times \mathbb{E}[V_T(S+r-v+R^*, R^{**},0)]\\
=& \sum_{R^*\notin\{S+r\}} p_{R^*}\times p_S(B)\times (\mathbb{E}[V_T(S+r-u+R^*, R^{**},0)]\\& -\mathbb{E}[V_T(S+r-v+R^*, R^{**},0)] \\
&+ (p_u-p_v)*p_S(B)\times \mathbb{E}[V_T(S+r, R^{**},0)]
\end{align*}
\normalsize
which is $\geq 0$ by induction, since\\
\begin{footnotesize}
$\mathbb{E}(V_T(S+r-u+R^*, R^{**},0))]\geq \mathbb{E}(V_T(S+r-v+R^*, R^{**},0))]$
\end{footnotesize}
\normalsize
and $p_u\geq p_v$ by definition of $v$. 
Similarly, we show that $\ref{45}(u)-\ref{45}(v)\geq 0$.
Finally, $\ref{46}(u)-\ref{47}(v)$ becomes:
\footnotesize
\begin{align*}
& \mathbb{E}[1_{R^*\notin\{S+r-u\},|S+r-u|=B} \widehat{V}_T{(S+r-u,R^*,0)}]\\& - \mathbb{E}[1_{R^*\notin\{S+r-v\},|S+r-v|=B} \widehat{V}_T{(S+r-v,R^*,0)}]\\
=& \sum_{R^*\notin\{S+r-u\}}P[R^*||S|=B]\times P[|S|=B]\widehat{V}_T{(S+r-u,R^*,0)}]\\
&-\sum_{R^*\notin\{S+r-v\}}P[R^*||S|=B] P[|S|=B]\widehat{V}_T{(S+r-v,R^*,0)}]\\
=& \sum_{R^*\notin\{S+r\}} (1-p_S(B))p_{R^*}\\
& (\widehat{V}_T{(S+r-u,R^*,0)}-\widehat{V}_T{(S+r-v,R^*,0)}) \\
&+ (1-p_S(B))(p_u\widehat{V}_T{(S+r-u,u,0)} -p_v\widehat{V}_T{(S+r-v,v,0)})
\end{align*} 
\normalsize
The first term in the above expression is always $\geq 0$ by induction. Now, we only need to show that the following sum 
%of the second term from the above expression and Equation \ref{48} 
is non-negative by using the definition of $\widehat{V}_T(.)$:
\small
\begin{align*}
& p_v\mathbb{E}[{V_T(S+r-u,R^{**},0)}]-p_u\mathbb{E}[{V_T(S+r-v,R^{**},0)}]\\+&(1-p_S(B))[p_u\mathbb{E}(V_T(S+r-v,R^{**},0))\\ & -p_v\mathbb{E}[V_T(S+r-v,R^{**},0)]]\\
=& p_v\mathbb{E}[{V_T(S+r-u,R^{**},0)}]\\& -[(1-p_S(B))(p_v-p_u)-p_u]\mathbb{E}[{V_T(S+r-v,R^{**},0)}]
\end{align*}
\normalsize
Now $\mathrm{\mathbb{E}[{V_T(S+r-v,R^{**},0)}]\leq \mathbb{E}[{V_T(S+r-u,R^{**},0)}]}$,by definition of $v$ therefore the above expression is non-negative if: 
$p_v\geq (1-p_S(B))(p_v-p_u)-p_u$
which is equivalent to showing $2p_u\geq p_S(B)(p_u-p_v)$. This holds since $2p_u=p_u+p_u\geq p_u-p_v \geq p_S(B)(p_u-p_v)$. This completes minimizing the first term.

% %
%\par We follow induction (as above) to show that minimizing the second term $\mathbb{E}_{D^*}[V_T(S+r-u,0,D^*)]$ reduces to minimizing the first term owing to (\ref{dpe}), (\ref{dpe2}) (see Appendix C of \cite{dropboxlink}).

% % TECHNICAL REPORT JHAMELA FOLLOWS
\textbf{Proving the result for the second minimization term:}\\
We now prove the result for the expression $\min_{u\in S+r}\mathbb{E}_{D^*}[V_T(S+r-u,0,D^*)]$. 
Similar to the case above, let
$v=\arg \min\{u\in S+r: p_uc(u)\}$, then, $\mathbb{E}[V_T(S+r-v,0,D^*)]=\min_{u\in S+r} \mathbb{E}[V_T(S+r-u,0,D^*)].$
The proof proceeds by induction on $T=1,2\dots$. At each step, we want to show that 
$\mathbb{E}[V_{T}(S+r-u,0,D^*)]-\mathbb{E}[V_{T}(S+r-v,0,D^*)]\geq 0.$\\
\textbf{The basis step}: Fix the initial state as $S$ and  departing file as $d$. Note that,
\begin{align*}
V_1(S,0,d)=&\mathbb{E}(V_0(S-d, R^*, 0)+\mathbb{E}(V_0(S-d,0,D^*))\\
=& 1_{\{R^*\notin {S-d}, |S-d|<B\}}2c(R^*)+\\
& 1_{\{R^*\notin {S-d},|S-d|=B\}}2 c(R^*)
\end{align*} 
This is because the second term in the above equation, i.e., $\mathbb{E}(V_0(S-d,0,D^*))=0$ by (\ref{dpe2}). Therefore, 
\begin{align*}
&\mathbb{E}[V_1(S+r-u,0,D^*)]\\=& \mathbb{E}[(1_{R^*\notin\{S+r-u-D^*\},|S+r-u-D^*|<B}+\\& 1_{R^*\notin\{S+r-u-D^*\},|S+r-u-D^*|=B})*2c(R^*)]\\
=& \mathbb{E}[1_{R^*\notin\{S+r-u-D^*\}}*2c(R^*)]\\
=& \mathbb{E}[1_{R^*\notin\{S+r\}}*2c(R^*)] +\\& \mathbb{E}[1_{R^*=u}*2c(R^*)]+\mathbb{E}[1_{R^*=D^*}*2c(R^*)]
\end{align*}
Now we write a similar expression for $\mathbb{E}[V_0(S+r-v,0,D^*)]$, and observe that,
\begin{align*}
&\mathbb{E}[V_1(S+r-u,0,D^*)]-\mathbb{E}[V_1(S+r-v,0,D^*)]\\
& =\mathbb{E}[1_{R^*\notin\{S+r\}}*2c(R^*)] + \mathbb{E}[1_{R^*=u}*2c(R^*)]\\
& +\mathbb{E}[1_{R^*=D^*}*2c(R^*)]- \mathbb{E}[1_{R^*\notin\{S+r\}}*2c(R^*)]\\& - \mathbb{E}[1_{R^*=v}*2c(R^*)]-\mathbb{E}[1_{R^*=D^*}*2c(R^*)]\\
& = \mathbb{E}[1_{R^*=u}*2c(R^*)]- \mathbb{E}[1_{R^*=v}*2c(R^*)]\\
& = 2(p_uc(u)-p_vc(v))
\end{align*}
That is, the claim is true with $T=1$.\\
\textbf{The induction step}:
Assume that the claim is true for some fixed $T>1$.
Fix $(S,r,0)$ and $(S,0,d)$ with $r\notin S$ and $d\in S$. We need to show that, for $u\in S+r$, we have, 
\begin{align}
\mathbb{E}[V_{T+1}(S+r-u,0,D^*)]-\mathbb{E}[V_{T+1}(S+r-v,0,D^*)]\geq 0 \nonumber
\end{align}
Now, we take the expectation of $V_{T+1}(S+r-u,0,D^*)$ and use Equation \ref{dpe2} to get,
\begin{align}
&\mathbb{E}_{D^*}[{V_{T+1}(S+r-u,0,D^*)}] \nonumber \\&=\mathbb{E}_{D^*}[\mathbb{E}_{R^*}(V_T(S+r-u-D^*,R^*,0))] \nonumber\\
&\; +\mathbb{E}_{D^*}[\mathbb{E}_{D^{**}}(V_T(S+r-u-D^*,0,D^{**}))] 
\label{similarity}
\end{align}
Note that the argument inside $\mathbb{E}_{D^*}[.]$ of the first term in equation \ref{similarity} is the same as solving for the induction step for the minimization of $\min_{u\in S+r}\mathbb{E}_{R^*}[V_T(S+r-u,R^*,0)]$. Also, the second term with double expectations on $D^*$ and $D^{**}$ would recur to solving $\min_{u\in S+r}\mathbb{E}_{D^*}[V_T(S+r-u,0,D^*)]$ again by the virtue of the recurrence equations \ref{dpe} and \ref{dpe2}. And we have already shown that the induction step holds for $\min_{u\in S+r}\mathbb{E}_{D^*}[V_T(S+r-u,0,D^*)]$. Hence, the claim.
This completes the proof of Theorem \ref{main-theorem}.

% %
\end{document}